\documentclass[12pt]{article}

\usepackage[a4paper,top=2.5cm,bottom=2.5cm,left=2.5cm,right=2.5cm,marginparwidth=1.75cm]{geometry}
\usepackage{lineno}
\usepackage{natbib}
\usepackage{amssymb}
\usepackage{bm}
\usepackage{authblk}
\usepackage{amsmath}
\usepackage{amsthm}
\usepackage{epsfig}
\usepackage{graphicx}
\usepackage{graphics}
\usepackage{float}
\usepackage{subfigure}
\usepackage{multirow}
\usepackage{color}
\usepackage{fullpage}
\usepackage[normalem]{ulem} 
\usepackage{makeidx}
\usepackage{xspace}
\usepackage{wrapfig}
\usepackage{setspace}
\usepackage{kotex}
\usepackage{url}
\usepackage{makecell}  
\makeindex

\newtheorem{theorem}{Theorem}

\newtheorem{remark}{Remark}

\newcommand{\thetab}{{\bm{\theta}}}

\providecommand{\keywords}[1]
{
  \small	
  \textbf{\textit{Keywords: }} #1
}

\begin{document}

\title{Differential gene expression analysis via two-component mixture models with a semiparametric skew-normal scale mixture alternative}

\author[1]{Sangkon Oh}
\affil[1]{Department of Statistics and Data Science, Pukyong National University}
\affil[1]{E-mail: ohsangkon@pknu.ac.kr}


\author[2*]{Geoffrey J. McLachlan}
\affil[2*]{School of Mathematics and Physics, University of Queensland}
\affil[2*]{E-mail: g.mclachlan@uq.edu.au}

\date{}

\maketitle

\begin{abstract}
Two–component mixture models are particularly useful for identifying differentially expressed genes, but their performance can deteriorate markedly when the alternative distribution departs from parametric assumptions or symmetry. We propose a semiparametric mixture model in which the null component is standard normal and the alternative follows a skew–normal scale mixture with an unspecified scale mixing distribution. This formulation accommodates skewness and heavy tails, providing a flexible and computationally tractable tool for differential gene–expression analysis without restrictive distributional assumptions. We establish identifiability and consistency of the model and develop an efficient estimation algorithm that incorporates nonparametric maximum likelihood estimation of the scale distribution. Numerical studies show notable improvements over existing parametric and nonparametric approaches for modeling the alternative distribution, and applications to colon cancer and leukemia datasets demonstrate reduced false discovery and false negative rates.
\end{abstract}

\keywords{Differential gene expression, Microarray data, Nonparametric maximum likelihood estimation, Semiparametric mixture models}

\section{Introduction} \label{sec:intro}


Gene expression analysis plays a critical role in elucidating the biological mechanisms underlying complex diseases, identifying key biomarker genes, and facilitating the development of diagnostic, preventive, and therapeutic strategies. In this context, the identification of differentially expressed genes has become a cornerstone of modern genomics, particularly with the advent of high-throughput technologies such as microarrays and RNA sequencing. These platforms enable the simultaneous quantification of expression levels for tens of thousands of genes, but they also introduce substantial statistical challenges due to the scale and multiplicity of hypothesis testing involved.

To address this issue, a wide array of statistical methods has been developed, among which penalized variable-selection techniques have gained significant prominence. Since the introduction of the least absolute shrinkage and selection operator (LASSO) by \citet{tibshirani1996regression}, numerous extensions and domain-specific applications have been proposed for gene expression analysis, including those by \citet{ghosh2005classification}, \citet{zhang2006gene}, \citet{ma2007supervised}, and \citet{guo2017l1}. These approaches are effective in managing high-dimensional data by promoting sparsity and identifying informative gene subsets. However, they also exhibit several limitations, as the selected features may not generalize reliably across models, and the repeated model fitting required in ultra-high-dimensional settings can lead to substantial computational burdens.

Alternatively, motivated in part by the two-component Gaussian mixture model for observed log ratios proposed by \citet{lee2000importance}, \citet{efron2001empirical} and \citet{efron2002empirical} introduced the local false discovery rate (FDR), a refinement of the classical tail-area–based FDR originally developed by \citet{benjamini1995controlling}. Subsequently, \citet{efron2004large} and \citet{efron2005local} proposed transforming observed statistics into $z$-scores to facilitate the application of local FDR methods for identifying differentially expressed genes. Under the assumption that gene expression levels are normally distributed and identically distributed for infected and uninfected cells, the $z$-scores corresponding to non-differentially expressed genes follow a standard normal distribution.

Building upon these ideas, \citet{mclachlan2006simple} introduced a two-component Gaussian mixture model that incorporates both $z$-score transformation and local FDR estimation. In this framework, the $z$-score distribution is modeled as
\begin{align}
g(z) = \pi f_0(z) + (1 - \pi) f_1(z),
\label{twocomp}
\end{align}
where $f_0(z)$ denotes the null distribution corresponding to non-differentially expressed genes, $f_1(z)$ represents the alternative distribution for differentially expressed genes, and $\pi$ is the proportion of genes that are not differentially expressed. The null distribution $f_0(z)$ is typically known or assumed to be standard normal, while the alternative distribution $f_1(z)$ is also modeled parametrically as a normal distribution. This formulation enables probabilistic inference at the gene level and provides a coherent basis for ranking genes according to their posterior probabilities of being differentially expressed.

Despite their computational convenience, parametric assumptions for the alternative distribution $f_1(z)$ in \eqref{twocomp} have inherent limitations, particularly when the true distribution deviates from the assumed form. To address this issue, \citet{bordes2006semiparametric} proposed a semiparametric two-component mixture model in which the null component is specified, while the alternative component is left completely unspecified but constrained to be symmetric. Their work established key identifiability results under symmetry assumptions and developed consistent estimation procedures based on moment equations and symmetrization techniques. Subsequent studies, including \citet{ma2015flexible}, \citet{pommeret2019semiparametric}, and \citet{milhaud2022semiparametric}, further extended this semiparametric framework. These contributions enhanced modeling flexibility by relaxing rigid parametric assumptions on the alternative distribution while maintaining identifiability.

However, the assumption that the alternative density is symmetric can still be overly restrictive in practice. To overcome this limitation, we propose a new semiparametric mixture modeling framework in which the null distribution is parametrically specified, while the alternative distribution is modeled using a semiparametric skew–normal scale mixture (SNSM) distribution. This semiparametric SNSM formulation provides enhanced flexibility by accommodating a broad class of distributions—including both symmetric and asymmetric unimodal forms—without requiring a model selection procedure. Estimation proceeds via a semiparametric maximum likelihood method, where the parameters of the skew-normal component are treated parametrically, and the mixing distribution is estimated nonparametrically using the nonparametric maximum likelihood estimator (NPMLE) framework of \citet{Lindsay_1995}.
Although not all distributions strictly fall within the class of skew-normal scale mixtures, the flexibility of the proposed semiparametric SNSM model is expected to mitigate potential issues arising from model misspecification. Recent studies, including \citet{lee2024finite}, \citet{seo2024adaptive}, and \citet{Oh2026}, have successfully applied the semiparametric SNSM framework in various contexts—such as finite mixture models, robust regression, and accelerated failure time modeling—demonstrating its effectiveness in capturing diverse distributional characteristics without imposing rigid parametric assumptions.

The remainder of this paper is organized as follows. Section~2 reviews the preprocessing steps used to obtain $z$-scores from microarray data, summarizes existing two-component mixture approaches for differential expression analysis, and introduces the semiparametric SNSM distribution. Section~3 presents the proposed semiparametric two-component mixture model—featuring a semiparametric SNSM alternative that accommodates both symmetric and asymmetric distributions—establishes identifiability and consistency, and details the estimation procedure. Section~4 reports a comprehensive simulation study comparing the proposed method with existing competitors. Section~5 applies the methodology to real gene-expression datasets. Finally, Section~6 concludes with a discussion of future research directions.

\section{Literature Review} 
\subsection{Data Preprocessing} \label{sec:preprocessing}

Let $\mathbf{X} = (x_{ij})_{N \times M}$ denote the gene expression matrix, where $N$ is the number of genes and $M$ is the number of tissue samples. Each row $\mathbf{x}_{i \cdot} = (x_{i1}, \ldots, x_{iM})$ represents the expression profile of gene $i$, and each sample $j$ belongs to one of $G$ experimental groups, labeled by $g_j \in \{1, \ldots, G\}$. For each gene $i$, we test the null hypothesis that its mean expression level does not differ across groups.

When $G = 2$, the equality of means can be assessed using the pooled two-sample $t$-statistic:
\[
t_i = 
\frac{\bar{x}_{i1} - \bar{x}_{i2}}
{s_{ip} \sqrt{1/m_1 + 1/m_2}},
\qquad
s_{ip}^2 = 
\frac{(m_1 - 1)s_{i1}^2 + (m_2 - 1)s_{i2}^2}{m_1 + m_2 - 2},
\]
where $\bar{x}_{ig}$ and $s_{ig}^2$ denote the sample mean and variance of gene $i$ in group $g$ ($g = 1, 2$), and $m_g$ is the number of samples in group $g$ (so that $m_1 + m_2 = M$). Under the null hypothesis, $t_i$ follows a $t$-distribution with $\nu = m_1 + m_2 - 2$ degrees of freedom. For multi-group settings ($G > 2$), the one-way ANOVA $F$-statistic is used analogously to test for equality of group means.

As the number of hypothesis tests increases, the risk of type I errors rises sharply. Therefore, it is essential to perform simultaneous inference that accounts for multiple testing. \citet{efron2004large}, \citet{efron2005local}, and \citet{mclachlan2006simple} proposed transforming each observed test statistic $t_i$ into a standard normal variable $Z_i$ to enable unified modeling across all genes. Let $F_0(\cdot)$ denote the cumulative distribution function (CDF) of $t_i$ under the null hypothesis. Based on $t_i$, \citet{mclachlan2006simple} proposed computing the two-sided $p$-value as
\begin{align}
P_i = 1 - F_0(t_i) + F_0(-t_i),
\label{eq:pval2sided}
\end{align}
which is then transformed into a one-sided upper-tail $z$-score via the probit (Gaussian quantile) transformation:
\begin{equation}
z_i = \Phi^{-1}(1 - P_i),
\label{eq:ztransform}
\end{equation}
where $\Phi^{-1}(\cdot)$ denotes the quantile function of the standard normal distribution. Under the assumption that gene expression levels are normally and identically distributed across groups, the $z$-scores corresponding to non-differentially expressed genes follow a standard normal distribution. That is, under the true null hypothesis, $z_i \sim \mathcal{N}(0, 1)$, while larger positive values of $z_i$ indicate stronger evidence of differential expression for gene $i$.

\subsection{Two-Component Mixture Model}

The transformed $z$-scores provide a standardized scale for assessing statistical evidence of differential expression. To move beyond conventional thresholding of $p$-values or $z$-scores, \citet{efron2001empirical} and \citet{efron2002empirical} proposed modeling the empirical distribution of test statistics using a two-component mixture model as in \eqref{twocomp}. This framework enables the computation of posterior probabilities, such as the local FDR, defined as the conditional probability that a gene is null given its observed $z$-score:
\begin{align}
\label{localFDR}
\gamma(z) = \frac{\pi f_0(z)}{g(z)},    
\end{align}
where $g(z)$ is the marginal density of $z$, $f_0(z)$ is the null density (typically standard normal), and $\pi$ is the prior probability that a gene is not differentially expressed.

A widely used specification for $g(z)$ is the two-component Gaussian mixture model proposed by \citet{mclachlan2006simple}:
\begin{align}
\label{gmm}
g(z) = \pi \phi(z; 0,1) + (1 - \pi)\phi(z; \mu_1, \sigma_1^2),   
\end{align}
where $\phi(z; \mu, \sigma^2)$ denotes the normal density with mean $\mu$ and variance $\sigma^2$, and the alternative distribution is characterized by $\mu_1 > 0$ to reflect right-tailed enrichment. While this Gaussian specification is computationally convenient, it imposes strong parametric assumptions on the alternative distribution $f_1(z)$, which is typically unknown and may deviate from normality in real biological data. Such misspecification can bias local FDR estimation and compromise inference accuracy.

To address these limitations, \citet{bordes2006semiparametric} introduced a semiparametric two-component mixture model in which the null distribution $f_0(z)$ is specified, while the alternative distribution $f_1(z)$ is left unspecified but assumed to be symmetric about an unknown location. This symmetry assumption is mild yet sufficient to ensure identifiability of the mixture components, and allows for modeling a broader range of symmetric but non-Gaussian distributions. Estimation methods based on projection or moment equations can be applied to consistently recover the mixture structure without fully specifying $f_1(z)$.

Recent extensions have further broadened the applicability of this framework. For instance, \citet{pommeret2019semiparametric} and \citet{milhaud2022semiparametric} generalized the semiparametric mixture model to two-sample settings, allowing for direct comparison between two unknown alternative distributions. These developments reinforce the semiparametric mixture approach as a powerful and interpretable alternative to both fully parametric and nonparametric models, particularly in large-scale gene expression analysis where the null distribution is well understood but the alternative remains elusive.

\subsection{Semiparametric Skew-Normal Scale Mixture Distribution}

In the semiparametric two-component mixture model considered in this work, the null density $f_0(z)$ is well represented by a standard normal distribution, reflecting the theoretical behavior of $z$-scores under the null hypothesis. In contrast, the alternative density $f_1(z)$ typically departs from normality, exhibiting features such as skewness, heavy tails, or scale heterogeneity. To flexibly accommodate these characteristics, we model $f_1(z)$ using a semiparametric SNSM distribution, which provides a unified framework capable of representing both symmetric and asymmetric families as well as a wide range of tail behaviors.

The starting point is the skew-normal distribution introduced by \citet{azzalini1985class}. For a location parameter $\mu$, scale parameter $\sigma > 0$, and shape parameter $\lambda$, its density is given by
\begin{equation*}
f_{\mathrm{SN}}(z; \mu, \lambda, \sigma)
= \frac{2}{\sigma}\,
\phi\!\left( \frac{z - \mu}{\sigma} \right)
\Phi\!\left( \lambda \frac{z - \mu}{\sigma} \right),
\end{equation*}
where $\phi(\cdot)$ and $\Phi(\cdot)$ denote the standard normal probability density and cumulative distribution functions, respectively. The shape parameter $\lambda$ controls the magnitude and direction of skewness, and the model reduces to the Gaussian distribution when $\lambda = 0$.

To extend this distribution toward heavier tails, \citet{branco2001general} introduced the SNSM family. Let $G$ be a mixing distribution on $(0,\infty)$; then the SNSM density takes the form
\begin{equation}
f_{\mathrm{SNSM}}(z; \mu, \lambda, G)
= \int
\frac{2}{\sigma}\,
\phi\!\left( \frac{z - \mu}{\sigma} \right)
\Phi\!\left( \lambda \frac{z - \mu}{\sigma} \right)
\, dG(\sigma),
\label{eq:ssnsm_density}
\end{equation}
which reduces to the skew-normal when $G$ is degenerate. Various skewed and heavy-tailed distributions arise as special cases for specific choices of $G$, including the skew-$t$, skew-slash, and skew-contaminated normal families. Moreover, setting $\lambda = 0$ yields Gaussian scale mixtures, which encompass the $t$-distribution, Laplace, logistic, and other well-known heavy-tailed models \citep{andrews1974scale, efron1978broad, west1987scale}.

While many existing applications impose a parametric form on $G$ in \eqref{eq:ssnsm_density}, doing so risks model misspecification if the true distribution of scales deviates from the chosen parametric family. Motivated by this, we treat $G$ as an infinite-dimensional parameter and estimate it nonparametrically. Let $z_1,\ldots,z_N$ be independent observations from the SNSM distribution. For fixed $(\mu,\lambda)$, the log-likelihood function is
\[
\tilde{\ell}_N(G)
= \sum_{i=1}^N \log f_{\mathrm{SNSM}}(z_i; \mu, \lambda, G).
\]
Because $G$ ranges over all probability measures on $(0,\infty)$, maximizing $\tilde{\ell}_N(G)$ requires specialized tools from nonparametric mixture theory. Following the work of \citet{Lindsay_1995}, one characterizes the NPMLE via directional derivatives. The directional derivative of $\tilde{\ell}_N$ at a candidate mixing distribution $\hat{G}$ in the direction of a point mass $H_\sigma$ is
\[
D_{\hat{G}}(\sigma)
= \lim_{\alpha \to 0}
\frac{
\tilde{\ell}_N\!\left\{ (1-\alpha)\hat{G} + \alpha H_\sigma \right\}
- \tilde{\ell}_N(\hat{G})
}{\alpha}.
\]
The NPMLE $\hat{G}$ satisfies the necessary and sufficient optimality conditions
\begin{align}
\label{NPMLE conditions}
D_{\hat{G}}(\sigma) \le 0 \quad \text{for all } \sigma > 0,
\qquad
D_{\hat{G}}(\sigma^*) = 0 \quad \text{for all } \sigma^* \in S(\hat{G}),   
\end{align}
where $S(\hat{G})$ denotes the support of $\hat{G}$. 

Models based on the NPMLE framework have been successfully applied in a wide range of settings, including linear and robust regression, heteroscedastic modeling, survival analysis, and finite mixture modeling \citep{SL15, XYS16, Seo_2017, AFT_NGSM, seo2024adaptive, lee2024finite, oh2024semiparametric, park2025penalized, Oh2026}. By adapting automatically to unknown distributional features without imposing restrictive parametric assumptions, these models offer a flexible and robust approach to statistical modeling across diverse applications.

\section{Proposed Method} \label{sec:proposed}
\subsection{Model Setup} \label{subsec:proposed}

To flexibly accommodate a broad class of alternative distributions in~\eqref{twocomp}, including those exhibiting asymmetry or heavy-tailed behavior, we propose a semiparametric two-component mixture model. The null component is modeled by a standard normal density, whereas the alternative component is represented by a semiparametric SNSM, which provides a unified mechanism for capturing skewness and scale heterogeneity.

Let $Z$ denote the $z$-score obtained through the transformation in~\eqref{eq:ztransform}. We assume that the density of $Z$ admits the mixture representation as
\begin{equation}
\label{eq:proposed_model}
f(z)
= \pi \,\phi(z; 0,1)
  + (1 - \pi)\, f_{\mathrm{SNSM}}(z; \mu, \lambda, G),
\end{equation}
where
\begin{itemize}
    \item $\phi(z; 0, 1)$ is the standard normal density defining the null component;
    \item $\pi \in (0,1)$ is the prior proportion of null genes;
    \item $f_{\mathrm{SNSM}}(z; \mu, \lambda, G)$ denotes the SNSM alternative density;
    \item $\mu > 0$ is the location shift under the alternative;
    \item $\lambda \ge 0$ is the skewness parameter;
    \item $G$ is an unspecified mixing distribution over scale parameters $\sigma>0$.
\end{itemize}

Model~\eqref{eq:proposed_model} generalizes the classical Gaussian mixture model in~\eqref{gmm} by replacing its parametric normal alternative with a semiparametric SNSM distribution. This modification permits the alternative density to exhibit a wide range of shapes—including symmetric, skewed, light-tailed, and heavy-tailed behaviors—while retaining a parsimonious and interpretable framework through the finite-dimensional parameters $(\mu, \lambda)$ and the nonparametric scale mixing distribution $G$.
In this setting, because non-null genes tend to produce positive $z$-scores, 
imposing $\mu > 0$ and $\lambda \ge 0$ ensures that the alternative component 
appropriately reflects the expected right-tail weight and prevents sign-reversal ambiguity.

The standard normal distribution belongs to the SNSM family only in the degenerate case where $\mu = 0$, $\lambda = 0$, and $G = H_1$. 
Under the restriction $\mu > 0$, the SNSM family therefore excludes the standard normal density. 
Consequently, the null Gaussian component and the alternative SNSM component are structurally distinct, and the former cannot be represented within the latter. 
This structural distinction suggests that the null component should be uniquely determined in the mixture representation.

The identifiability of model~\eqref{eq:proposed_model} is not immediate due to the presence of the infinite-dimensional mixing distribution \(G\). To address this issue, we impose a set of regularity conditions that ensure a clear separation between the null and alternative components and prevent degeneracy in the scale mixture representation. In particular, Assumption~(A1) establishes linear independence between the null Gaussian density and the SNSM alternative class. Assumptions~(A2)--(A3) impose mild regularity conditions on the mixing distribution \(G\), ensuring that the scale parameter remains bounded away from zero and that the likelihood function is well-defined. These latter conditions are primarily required for establishing consistency of the maximum likelihood estimator and for guaranteeing the well-posedness of the estimation problem.

\medskip
\noindent\textbf{Assumptions.}
\begin{enumerate}
\item[(A1)] The null Gaussian density \(\phi(\cdot;0,1)\) is linearly independent of the non-null SNSM class
\[
\mathcal{F}_{\mathrm{SNSM}}^{+}
=
\left\{
f_{\mathrm{SNSM}}(\cdot;\mu,\lambda,G)
:
\mu>0,\ \lambda\ge 0,\ G \in \mathcal{G}
\right\}.
\]
That is, if
\[
a_0\,\phi(\cdot;0,1)
+
\sum_{j=1}^k a_j f_j(\cdot)
= 0
\quad \text{a.e.},
\]
for some \(k \in \mathbb{N}\), coefficients \(a_0,\dots,a_k \in \mathbb{R}\), and functions \(f_j \in \mathcal{F}_{\mathrm{SNSM}}^{+}\), then
\[
a_0 = a_1 = \cdots = a_k = 0.
\]

\item[(A2)] The support of \(G\) is contained in \([\ell, \infty)\) for some constant \(\ell > 0\).

\item[(A3)] The mixing distribution \(G\) satisfies
\[
\int_{\ell}^{\infty} \log \sigma \, dG(\sigma) < \infty.
\]
\end{enumerate}

The following theorem establishes identifiability of the parameters \((\pi,\mu,\lambda,G)\) under the separation condition in Assumption~(A1).

\begin{theorem}
\label{thm1}
Suppose Assumption~\textup{(A1)} holds. Let $\mu,\tilde{\mu}>0$, $\lambda,\tilde{\lambda}\ge 0$, and $G,\tilde{G}$ be probability measures on $(0,\infty)$. If
\[
\pi \phi(z;0,1) + (1-\pi)f_{\mathrm{SNSM}}(z;\mu,\lambda,G)
=
\tilde{\pi}\phi(z;0,1) + (1-\tilde{\pi})f_{\mathrm{SNSM}}(z;\tilde{\mu},\tilde{\lambda},\tilde{G})
\]
for almost all $z\in\mathbb{R}$, then
\[
(\pi,\mu,\lambda,G)=(\tilde{\pi},\tilde{\mu},\tilde{\lambda},\tilde{G}).
\]
\end{theorem}

\begin{proof}
A proof is given in Appendix~A.
\end{proof}

\begin{remark}
When $\lambda=0$ and $G$ is degenerate at a fixed scale value, model~\eqref{eq:proposed_model} reduces to the standard two-component Gaussian mixture model with distinct means, as given in~\eqref{gmm}.  
Thus, Theorem~\ref{thm1} encompasses the classical two-component Gaussian mixture as a special case.
\end{remark}

\begin{remark}
Because the $z$-scores in~\eqref{eq:ztransform} assign larger positive values to stronger evidence against the null, we restrict attention to the right-tailed setting with $\mu>0$. More generally, left-tailed alternatives can also be accommodated by allowing $\mu<0$ and aligning the direction of skewness through the constraint 
$\lambda \ge 0$ when $\mu>0$ and $\lambda \le 0$ when $\mu<0$.
\end{remark}

The general consistency of maximum likelihood estimators in semiparametric mixture models was demonstrated by \citet{kiefer1956consistency}.  
By verifying their regularity conditions for model~\eqref{eq:proposed_model}, we obtain the following consistency result for the estimators of $(\pi,\mu,\lambda,G)$.

\begin{theorem}
\label{thm2}
Suppose that Assumptions~\textup{(A1)}--\textup{(A3)} hold. Then the MLE \((\hat{\pi}, \hat{\mu}, \hat{\lambda}, \hat{G})^\top\) is consistent for \((\pi, \mu, \lambda, G)^\top\).
\end{theorem}

\begin{proof}
A proof is given in the Appendix B. 
\end{proof}

\subsection{Expectation–Conditional Maximization Algorithm} \label{sec:em}

Let $z_i \in \mathbb{R}$ denote the $z$-score for gene $i$, for $i=1,\dots,N$.  
The observed log-likelihood corresponding to the two-component mixture model in~\eqref{eq:proposed_model} is
\begin{align}
\label{likelihood}
    \ell(\thetab,G)
    \;=\;
    \sum_{i=1}^{N}
    \log \Big\{
        \pi\,\phi(z_i;0,1)
        + (1-\pi)\, f_{\mathrm{SNSM}}(z_i;\mu,\lambda,G)
    \Big\},
\end{align}
where $\thetab=(\pi,\mu,\lambda)$.  
Maximizing~\eqref{likelihood} is challenging due to the infinite-dimensional parameter $G$.  
We therefore adopt the expectation–conditional maximization (ECM) algorithm of \citet{meng1993maximum}.

Introduce the latent variable $\gamma\in\{0,1\}$, where $\gamma=1$ if $Z$ arises from the null component and $\gamma=0$ otherwise.  
The complete log-likelihood is then
\begin{align}
\label{complete}
\ell_c(\thetab,G)
=\sum_{i=1}^{N}
    \gamma_i \log\!\left\{ \pi\,\phi(z_i;0,1) \right\}
  +\sum_{i=1}^{N}
    (1-\gamma_i)\log\!\left\{ (1-\pi)\,f_{\mathrm{SNSM}}(z_i;\mu,\lambda,G) \right\}.
\end{align}

\paragraph{E-step.}
Given current estimates $(\thetab^{(t)},G^{(t)})$, the conditional expectation of~\eqref{complete} is
\[
Q(\thetab,G \mid \thetab^{(t)},G^{(t)})
=
\mathbb{E}\!\left[\ell_c(\thetab,G)\mid z_1,\dots,z_N,\thetab^{(t)},G^{(t)}\right],
\]
where the posterior probability that observation $z_i$ belongs to the null component is
\[
\gamma_i^{(t)}
=\Pr(U_i=1\mid z_i,\thetab^{(t)},G^{(t)})
=
\frac{
\pi^{(t)} \phi(z_i;0,1)
}{
\pi^{(t)}\phi(z_i;0,1)
+
(1-\pi^{(t)})\!
\displaystyle\int
\frac{2}{\sigma}
\phi\!\Big(\tfrac{z_i-\mu^{(t)}}{\sigma}\Big)
\Phi\!\Big(\lambda^{(t)}\tfrac{z_i-\mu^{(t)}}{\sigma}\Big)
\, dG^{(t)}(\sigma)
}.
\]

Substituting $\gamma_i^{(t)}$ into the expectation gives
\begin{align*}
Q(\thetab,G)
&=
\sum_{i=1}^{N}
\gamma_i^{(t)}\big\{\log\pi+\log\phi(z_i;0,1)\big\}
+
\sum_{i=1}^{N}
(1-\gamma_i^{(t)})\log(1-\pi)
\\
&\quad+
\sum_{i=1}^{N}
(1-\gamma_i^{(t)})\,
\log\!\left\{
\int
\frac{2}{\sigma}
\phi\!\Big(\tfrac{z_i-\mu}{\sigma}\Big)
\Phi\!\Big(\lambda\tfrac{z_i-\mu}{\sigma}\Big)
dG(\sigma)
\right\}.
\end{align*}

\paragraph{CM-step 1: Updating $G$ (NPMLE).}
Conditioning on $\thetab^{(t)}$, the update of $G$ requires maximizing
\[
\tilde{\ell}_N(G)
=
\sum_{i=1}^{N}
(1-\gamma_i^{(t)})
\log\!\left\{
\int
\frac{2}{\sigma}
\phi\!\Big(\tfrac{z_i-\mu^{(t)}}{\sigma}\Big)
\Phi\!\Big(\lambda^{(t)}\tfrac{z_i-\mu^{(t)}}{\sigma}\Big)
dG(\sigma)
\right\}.
\]
This is a convex but infinite-dimensional optimization problem.  
We employ directional-derivative based algorithms such as the vertex direction method \citep{Bohning85}, vertex exchange method \citep{Bohning86}, intra-simplex direction method \citep{LK92}, or constrained Newton method (CNM; \citealp{Wang07}).
The directional derivative at $G^{(t)}$ in the direction of $H_\sigma$ is
\[
D_{G^{(t)}}(\sigma)
=
\sum_{i=1}^{N}
(1-\gamma_i^{(t)})
\left[
\frac{
\dfrac{2}{\sigma}
\phi\!\big(\tfrac{z_i-\mu^{(t)}}{\sigma}\big)
\Phi\!\big(\lambda^{(t)}\tfrac{z_i-\mu^{(t)}}{\sigma}\big)
}{
\displaystyle\int
\frac{2}{s}
\phi\!\big(\tfrac{z_i-\mu^{(t)}}{s}\big)
\Phi\!\big(\lambda^{(t)}\tfrac{z_i-\mu^{(t)}}{s}\big)
dG^{(t)}(s)
}
\right]
-
\sum_{i=1}^{N}
(1-\gamma_i^{(t)}).
\]
Updates proceed along directions with $D_{G^{(t)}}(\sigma)>0$ until the NPMLE $G^{(t+1)}$ satisfies the optimality conditions of \citet{Lindsay_1995}.

\paragraph{CM-step 2: Updating $(\pi,\mu,\lambda)$.}
With $G^{(t+1)}$ fixed, the mixing proportion is updated as
\[
\pi^{(t+1)}
=
\frac{1}{N}\sum_{i=1}^{N}\gamma_i^{(t)}.
\]
The remaining parameters are updated by solving
\[
(\mu^{(t+1)},\lambda^{(t+1)})
=
\arg\max_{\mu,\lambda}
\sum_{i=1}^{N}
(1-\gamma_i^{(t)})
\log
\left\{
\int
\frac{2}{\sigma}
\phi\!\Big(\tfrac{z_i-\mu}{\sigma}\Big)
\Phi\!\Big(\lambda\tfrac{z_i-\mu}{\sigma}\Big)
dG^{(t+1)}(\sigma)
\right\},
\]
which can be efficiently performed using the Broyden–Fletcher–Goldfarb–Shanno (BFGS) algorithm  
\citep{broyden1970convergence, fletcher1970new, goldfarb1970family, shanno1970conditioning}.

\medskip\noindent
The E-step and the two CM-steps are iterated until convergence of the observed log-likelihood.

\section{Simulation Study} \label{sec:sim}

We conduct a comprehensive Monte Carlo simulation study to evaluate the performance of three approaches for distinguishing differentially expressed genes from non–differentially expressed ones:
(i) the parametric two-component Gaussian mixture model;
(ii) the nonparametric semiparametric two-component method, which models the alternative density nonparametrically under a symmetry constraint; and
(iii) the proposed semiparametric SNSM mixture model, which allows skewness and heavy-tailed behavior through an unspecified scale mixing distribution.
Throughout this section, these three approaches are referred to as the parametric, nonparametric, and semiparametric methods, respectively, reflecting the degree of assumptions placed on the alternative distribution.

The aim of the simulation study is to examine how these three methods perform under a broad range of alternative distributions differing in symmetry, skewness, and tail weight, and under varying degrees of sparsity in the non-null component and different total numbers of genes. Each synthetic dataset consists of independent observations
\[
Z_i \sim \pi\,\mathcal{N}(0,1) \;+\; (1-\pi)\,f_1, \qquad i=1,\dots,N,
\]
where the null component is standard normal and $f_1$ denotes the alternative (non-null) distribution.  
To capture a diverse set of realistic scenarios, we consider six representative choices for $f_1$:
\begin{align*}
\text{Case I:}\;& f_1=\mathcal{N}(\mu,1) 
&& \text{(symmetric, light-tailed)},\\[2pt]
\text{Case II:}\;& f_1=0.9\,\mathcal{N}(\mu,1)+0.1\,\mathcal{N}(\mu,2) 
&& \text{(symmetric, moderately heavy-tailed)},\\[2pt]
\text{Case III:}\;& f_1=\mathrm{Laplace}(\mu,1) 
&& \text{(symmetric, heavy-tailed)},\\[2pt]
\text{Case IV:}\;& f_1=\mathrm{t}(\mu,\nu=10)
&& \text{(symmetric, heavy-tailed)},\\[2pt]
\text{Case V:}\;& f_1=\mathrm{SN}(\mu,1,\lambda=5)
&& \text{(asymmetric)},\\[2pt]
\text{Case VI:}\;& f_1=\mathrm{Skew}\text{-}t(\mu,1,\lambda=5,\nu=10)
&& \text{(asymmetric and heavy-tailed)}.
\end{align*}
The alternative location is set to $\mu=1.645$, corresponding to the $95\%$ quantile of the standard normal distribution, providing a moderately separated but non-extreme signal strength on the z-scale.

To investigate the effects of sparsity and dimensionality, we vary the null proportion and the total number of genes as
\[
\pi \in \{0.3,\,0.5,\,0.7\}, 
\qquad 
N \in \{1000,\,5000\}.
\]
Larger values of $\pi$ correspond to sparser settings, with $\pi=0.7$ indicating that most genes follow the null distribution, while $\pi=0.3$ yields a comparatively dense non-null setting. For each combination of $(f_1,\pi,N)$, we generate 200 independent replications to obtain stable Monte Carlo performance measures. 

Because the true component membership is known in simulation, clustering accuracy is evaluated by classifying each observation as null or non-null according to its posterior probability, using the maximum a posteriori (MAP) rule with threshold $\gamma(z) = 0.5$. Performance is then quantified using the Adjusted Rand Index (ARI; \citealp{hubert1985comparing}) and Adjusted Mutual Information (AMI; \citealp{vinh2010information}), computed across the 200 replications. Both ARI and AMI correct for chance agreement, with higher values indicating better alignment between estimated and true classifications.

Across all simulation settings, the proposed semiparametric approach consistently attains the highest ARI and AMI values, demonstrating superior accuracy in recovering the true null and non-null classifications. Under symmetric and light-tailed alternatives (Case~I), the three methods perform similarly, though the proposed estimator retains a slight but systematic advantage. In symmetric yet heavy-tailed settings (Cases~II–IV), the parametric Gaussian mixture deteriorates substantially due to tail misspecification, and while the nonparametric approach improves upon the parametric method, it remains clearly inferior to the proposed estimator, which more effectively accounts for scale heterogeneity through its flexible mixing distribution. Under asymmetric alternatives (Cases~V–VI), the superiority of the proposed method is most striking, as neither the parametric nor the symmetric nonparametric model can accommodate skewness. Increasing the sample size from $N=1000$ to $N=5000$ improves the performance of all methods, but the semiparametric estimator continues to dominate, underscoring its robustness and adaptability across symmetric, asymmetric, light-tailed, and heavy-tailed alternatives.

\begin{table}[h]
\caption{ARI and AMI under each simulation scenario when $N=1000$ (Boldface indicates the best performance in each setting)}
\label{tab:p1000} 
\centering
\begin{tabular}{c c c c c c c c} 
\hline \hline \noalign{\smallskip}
\multirow{2}{*}{Case} &
\multirow{2}{*}{$\pi$} &
\multicolumn{2}{c}{Parametric method} &
\multicolumn{2}{c}{Nonparametric method} &
\multicolumn{2}{c}{Semiparametric method} \\
 & & ARI & AMI & ARI & AMI & ARI & AMI \\
\hline
\multirow{3}{*}{$\uppercase\expandafter{\romannumeral1}$} 
 & $0.3$ & 0.3717 & 0.2336 & 0.3649 & 0.2292 & $\boldsymbol{0.3757}$ & $\boldsymbol{0.2471}$ \\
 & $0.5$ & 0.3348 & 0.2627  & 0.3144 & 0.2463 & $\boldsymbol{0.3363}$ & $\boldsymbol{0.2642}$ \\
 & $0.7$ & 0.3799 & 0.2426  & 0.3522 & 0.2287  & $\boldsymbol{0.3901}$ & $\boldsymbol{0.2468}$ \\
\hline
\multirow{3}{*}{$\uppercase\expandafter{\romannumeral2}$} 
 & $0.3$ & 0.2036 & 0.1249 & 0.3178 & 0.1977 & $\boldsymbol{0.3569}$ & $\boldsymbol{0.2303}$ \\
 & $0.5$ & 0.1808 & 0.1493  & 0.2699 & 0.2132 & $\boldsymbol{0.3191}$ & $\boldsymbol{0.2521}$ \\
 & $0.7$ & 0.3426 & 0.2325  & 0.3476 & 0.2253  & $\boldsymbol{0.3852}$ & $\boldsymbol{0.2412}$ \\
\hline
\multirow{3}{*}{$\uppercase\expandafter{\romannumeral3}$} 
 & $0.3$ & 0.3896 & 0.2495 & 0.4260 & 0.2845 & $\boldsymbol{0.4345}$ & $\boldsymbol{0.2987}$ \\
 & $0.5$ & 0.3081 & 0.2473  & 0.3941 & 0.3095 & $\boldsymbol{0.4049}$ & $\boldsymbol{0.3180}$ \\
 & $0.7$ & 0.3053 & 0.2275  & 0.3950 & 0.2671  & $\boldsymbol{0.4390}$ & $\boldsymbol{0.2926}$ \\
\hline
\multirow{3}{*}{$\uppercase\expandafter{\romannumeral4}$} 
 & $0.3$ & 0.0035 & 0.0021 & 0.0659 & 0.0429 & $\boldsymbol{0.3579}$ & $\boldsymbol{0.2252}$ \\
 & $0.5$ & 0.1661 & 0.1581  & 0.2111 & 0.1821 & $\boldsymbol{0.3374}$ & $\boldsymbol{0.2644}$ \\
 & $0.7$ & 0.3778 & 0.2459  & 0.3442 & 0.2301  & $\boldsymbol{0.3868}$ & $\boldsymbol{0.2449}$ \\
\hline
\multirow{3}{*}{$\uppercase\expandafter{\romannumeral5}$} 
 & $0.3$ & 0.8287 & 0.7312 & 0.8534 & 0.7604 & $\boldsymbol{0.8805}$ & $\boldsymbol{0.7868}$ \\
 & $0.5$ & 0.7976 & 0.7434  & 0.7427 & 0.6898 & $\boldsymbol{0.8312}$ & $\boldsymbol{0.7567}$ \\
 & $0.7$ & 0.7903 & 0.6839  & 0.7467 & 0.6439  & $\boldsymbol{0.7964}$ & $\boldsymbol{0.6769}$ \\
\hline
\multirow{3}{*}{$\uppercase\expandafter{\romannumeral6}$} 
 & $0.3$ & 0.7892 & 0.6846 & 0.8401 & 0.7420 & $\boldsymbol{0.8753}$ & $\boldsymbol{0.7779}$ \\
 & $0.5$ & 0.7672 & 0.7158  & 0.7434 & 0.6882 & $\boldsymbol{0.8281}$ & $\boldsymbol{0.7510}$ \\
 & $0.7$ & 0.7798 & 0.6778  & 0.7299 & 0.6291  & $\boldsymbol{0.7986}$ & $\boldsymbol{0.6787}$ \\
\noalign{\smallskip}\hline\noalign{\smallskip}
\end{tabular}
\end{table}

\begin{table}[h]
\caption{ARI and AMI under each simulation scenario when $N=5000$ (Boldface indicates the best performance in each setting)}
\label{tab:p5000} 
\centering
\begin{tabular}{c c c c c c c c} 
\hline \hline \noalign{\smallskip}
\multirow{2}{*}{Case} &
\multirow{2}{*}{$\pi$} &
\multicolumn{2}{c}{Parametric method} &
\multicolumn{2}{c}{Nonparametric method} &
\multicolumn{2}{c}{Semiparametric method} \\
 & & ARI & AMI & ARI & AMI & ARI & AMI \\
\hline
\multirow{3}{*}{$\uppercase\expandafter{\romannumeral1}$} 
 & $0.3$ & $\boldsymbol{0.3907}$ & 0.2452 & 0.3789 & 0.2380 & 0.3869 & $\boldsymbol{0.2541}$ \\
 & $0.5$ & $\boldsymbol{0.3454}$ & $\boldsymbol{0.2666}$  & 0.3383 & 0.2612 & 0.3403 & 0.2650 \\
 & $0.7$ & 0.3918 & 0.2456  & 0.3470 & 0.2252  & $\boldsymbol{0.3926}$ & $\boldsymbol{0.2471}$ \\
\hline
\multirow{3}{*}{$\uppercase\expandafter{\romannumeral2}$} 
 & $0.3$ & 0.1369 & 0.0812 & 0.3503 & 0.2155 & $\boldsymbol{0.3640}$ & $\boldsymbol{0.2318}$ \\
 & $0.5$ & 0.1069 & 0.0864  & 0.2956 & 0.2280 & $\boldsymbol{0.3263}$ & $\boldsymbol{0.2531}$ \\
 & $0.7$ & 0.3385 & 0.2287  & 0.3585 & 0.2287  & $\boldsymbol{0.3820}$ & $\boldsymbol{0.2358}$ \\
\hline
\multirow{3}{*}{$\uppercase\expandafter{\romannumeral3}$} 
 & $0.3$ & 0.4079 & 0.2581 & 0.4273 & 0.2869 & $\boldsymbol{0.4328}$ & $\boldsymbol{0.2975}$ \\
 & $0.5$ & 0.3590 & 0.2870  & 0.3742 & 0.2935 & $\boldsymbol{0.4115}$ & $\boldsymbol{0.3224}$ \\
 & $0.7$ & 0.2201 & 0.1727  & 0.4256 & 0.2869  & $\boldsymbol{0.4403}$ & $\boldsymbol{0.2917}$ \\
\hline
\multirow{3}{*}{$\uppercase\expandafter{\romannumeral4}$} 
 & $0.3$ & 0.0000 & 0.0000 & 0.0202 & 0.0164 & $\boldsymbol{0.3841}$ & $\boldsymbol{0.2399}$ \\
 & $0.5$ & 0.1773 & 0.1845  & 0.2453 & 0.2164 & $\boldsymbol{0.3460}$ & $\boldsymbol{0.2673}$ \\
 & $0.7$ & $\boldsymbol{0.3933}$ & $\boldsymbol{0.2548}$  & 0.3774 & 0.2462  & 0.3931 & 0.2467 \\
\hline
\multirow{3}{*}{$\uppercase\expandafter{\romannumeral5}$} 
 & $0.3$ & 0.8271 & 0.7287 & 0.8687 & 0.7758 & $\boldsymbol{0.8804}$ & $\boldsymbol{0.7846}$ \\
 & $0.5$ & 0.7958 & 0.7402  & 0.7385 & 0.6805 & $\boldsymbol{0.8316}$ & $\boldsymbol{0.7551}$ \\
 & $0.7$ & 0.7911 & 0.6836  & 0.7443 & 0.6387  & $\boldsymbol{0.7984}$ & $\boldsymbol{0.6773}$ \\
\hline
\multirow{3}{*}{$\uppercase\expandafter{\romannumeral6}$} 
 & $0.3$ & 0.7937 & 0.6890 & 0.8600 & 0.7625 & $\boldsymbol{0.8766}$ & $\boldsymbol{0.7776}$ \\
 & $0.5$ & 0.7699 & 0.7169  & 0.7740 & 0.7112 & $\boldsymbol{0.8298}$ & $\boldsymbol{0.7507}$ \\
 & $0.7$ & 0.7798 & $\boldsymbol{0.6757}$  & 0.6308 & 0.5403  & $\boldsymbol{0.7971}$ & 0.6750 \\
\noalign{\smallskip}\hline\noalign{\smallskip}
\end{tabular}
\end{table}

Figures~1--6 present the estimated null, alternative, and mixture densities obtained from the parametric, nonparametric, and proposed semiparametric approaches across the six simulation scenarios. In the normal setting (Case~I), all three methods recover the overall mixture shape reasonably well. As the alternative distribution becomes heavier–tailed (Cases~II--IV), the parametric method exhibits pronounced distortions in the tails, reflecting its sensitivity to misspecification, while the nonparametric method improves upon the parametric fit but remains limited in accommodating tail heterogeneity. In contrast, the proposed method successfully adapts to variations in tail weight and shape by learning the underlying scale mixing distribution. Under asymmetric alternatives (Cases~V--VI), the limitations of the competing methods become more apparent: the parametric method enforces symmetry and fails to capture right–skewed structure, and the nonparametric method cannot represent directional asymmetry. The proposed semiparametric approach, however, accurately recovers both skewness and heavy–tailed behavior, providing the closest match to the true alternative density across all scenarios.

\begin{figure}[p] 
\centering

\subfigure[Parametric method]{
\includegraphics[width=0.31 \linewidth]{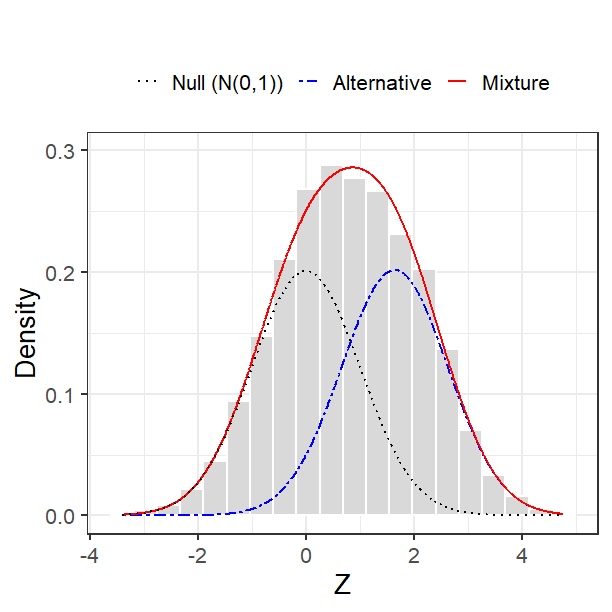}
}
\subfigure[Nonparametric method]{
\includegraphics[width=0.31 \linewidth]{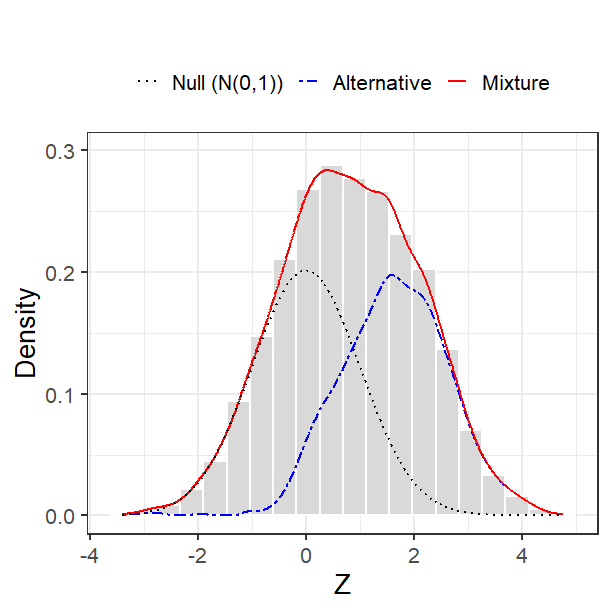}
}
\subfigure[Semiparametric method]{
\includegraphics[width=0.31 \linewidth]{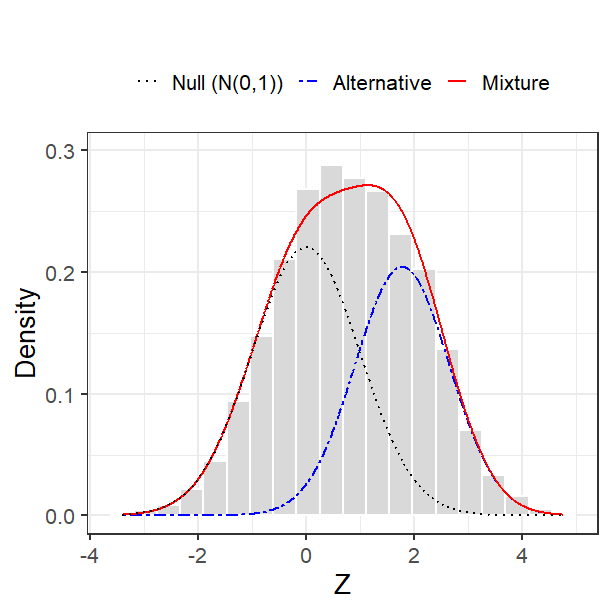}
}

\caption{Estimated density from each method with one simulated sample of Case $\uppercase\expandafter{\romannumeral1}$ when $N = 5000$}
\label{fig: Case 1}
\end{figure}

\begin{figure}[p] 
\centering

\subfigure[Parametric method]{
\includegraphics[width=0.31 \linewidth]{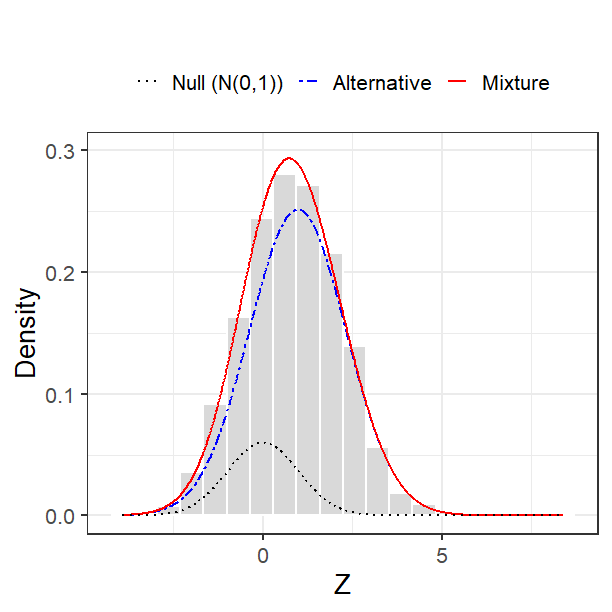}
}
\subfigure[Nonparametric method]{
\includegraphics[width=0.31 \linewidth]{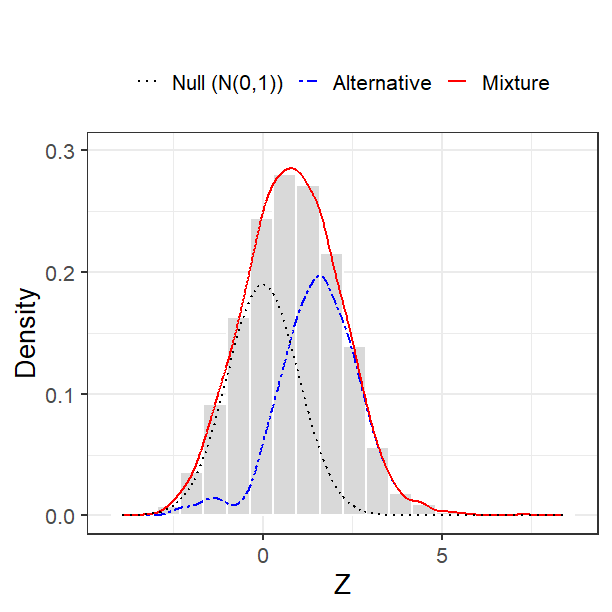}
}
\subfigure[Semiparametric method]{
\includegraphics[width=0.31 \linewidth]{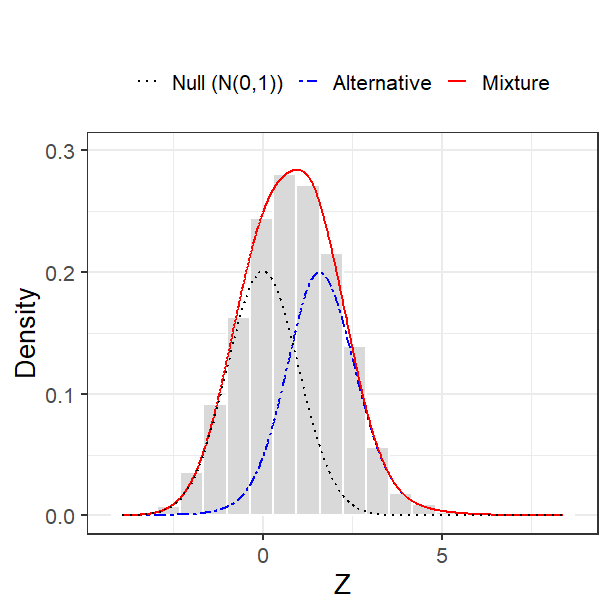}
}

\caption{Estimated density from each method with one simulated sample of Case $\uppercase\expandafter{\romannumeral2}$ when $N = 5000$}
\label{fig: Case 2}
\end{figure}

\begin{figure}[p] 
\centering

\subfigure[Parametric method]{
\includegraphics[width=0.31 \linewidth]{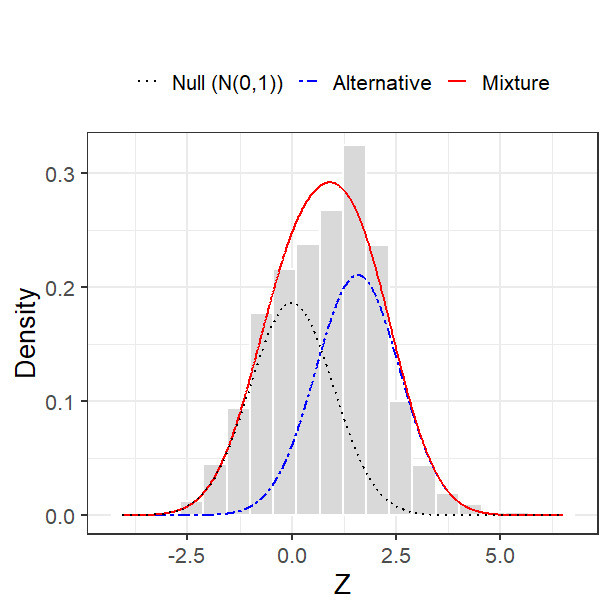}
}
\subfigure[Nonparametric method]{
\includegraphics[width=0.31 \linewidth]{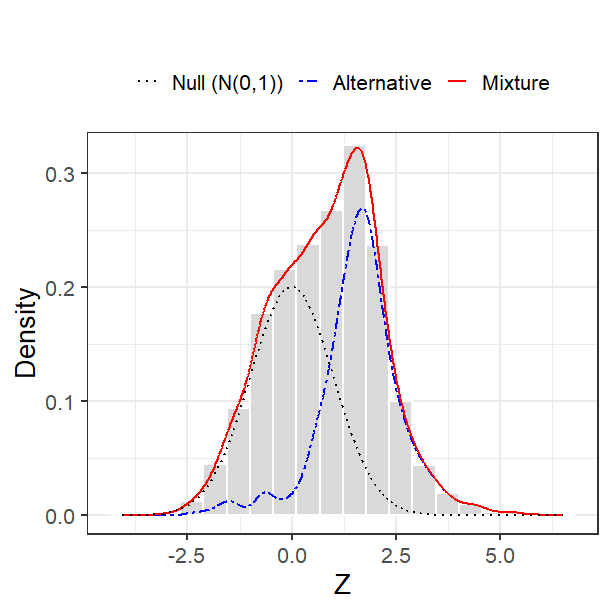}
}
\subfigure[Semiparametric method]{
\includegraphics[width=0.31 \linewidth]{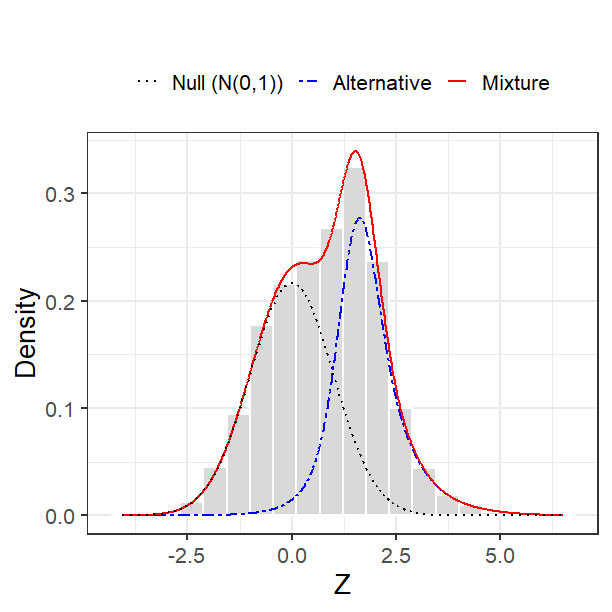}
}

\caption{Estimated density from each method with one simulated sample of Case $\uppercase\expandafter{\romannumeral3}$ when $N = 5000$}
\label{fig: Case 3}
\end{figure}

\begin{figure}[p] 
\centering

\subfigure[Parametric method]{
\includegraphics[width=0.31 \linewidth]{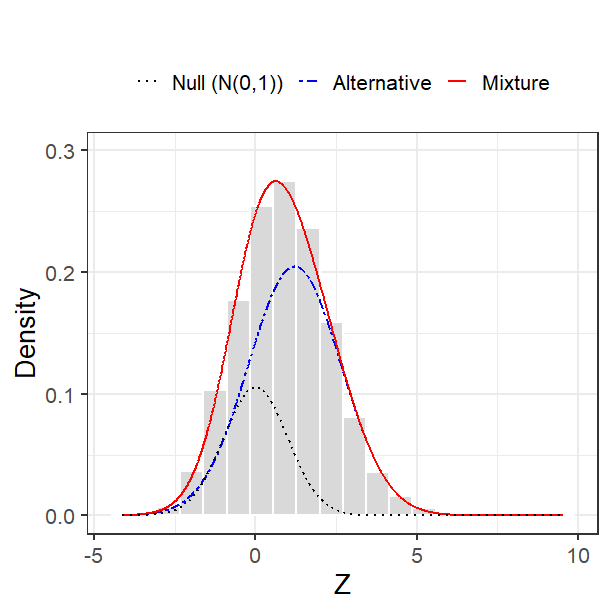}
}
\subfigure[Nonparametric method]{
\includegraphics[width=0.31 \linewidth]{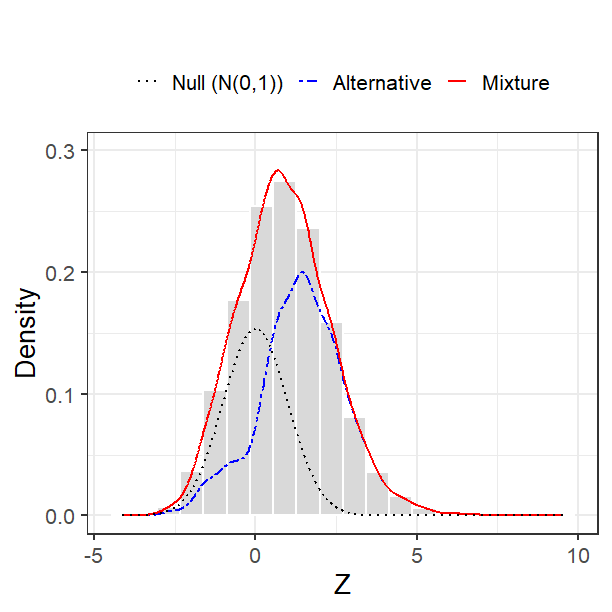}
}
\subfigure[Semiparametric method]{
\includegraphics[width=0.31 \linewidth]{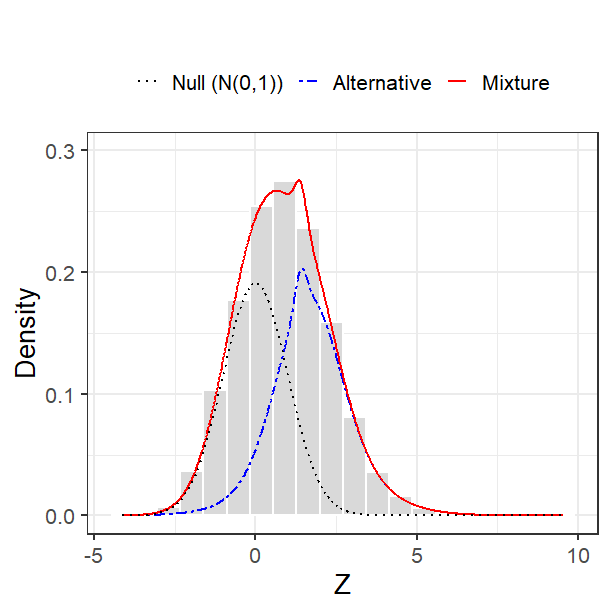}
}

\caption{Estimated density from each method with one simulated sample of Case $\uppercase\expandafter{\romannumeral4}$ when $N = 5000$}
\label{fig: Case 4}
\end{figure}

\begin{figure}[p] 
\centering

\subfigure[Parametric method]{
\includegraphics[width=0.31 \linewidth]{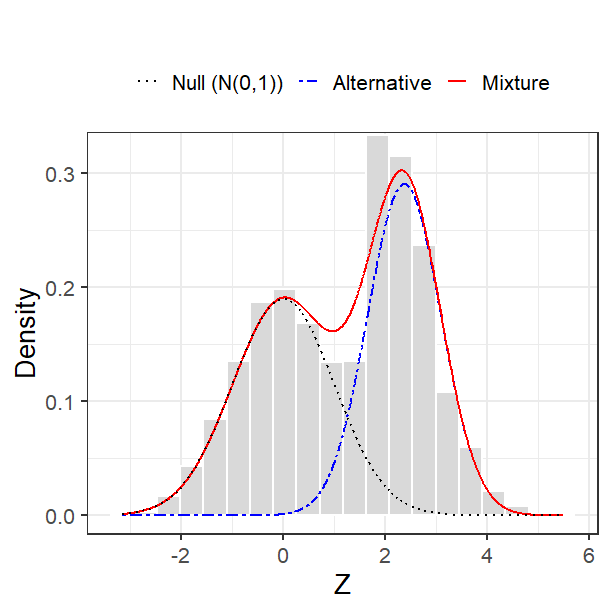}
}
\subfigure[Nonparametric method]{
\includegraphics[width=0.31 \linewidth]{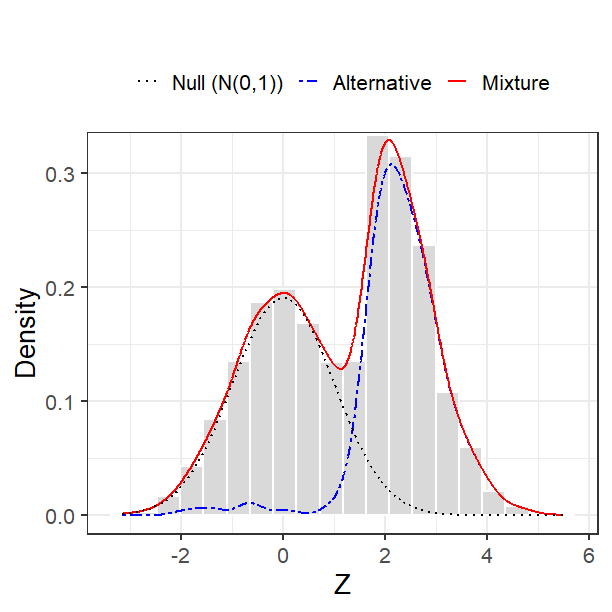}
}
\subfigure[Semiparametric method]{
\includegraphics[width=0.31 \linewidth]{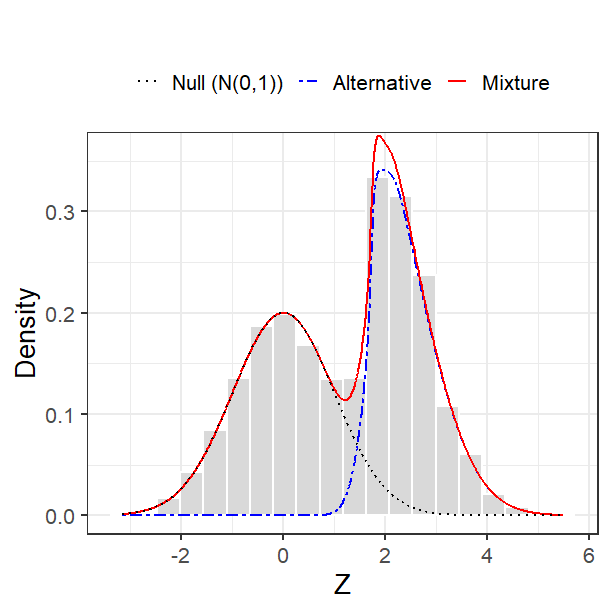}
}

\caption{Estimated density from each method with one simulated sample of Case $\uppercase\expandafter{\romannumeral5}$ when $N = 5000$}
\label{fig: Case 5}
\end{figure}

\begin{figure}[p] 
\centering

\subfigure[Parametric method]{
\includegraphics[width=0.31 \linewidth]{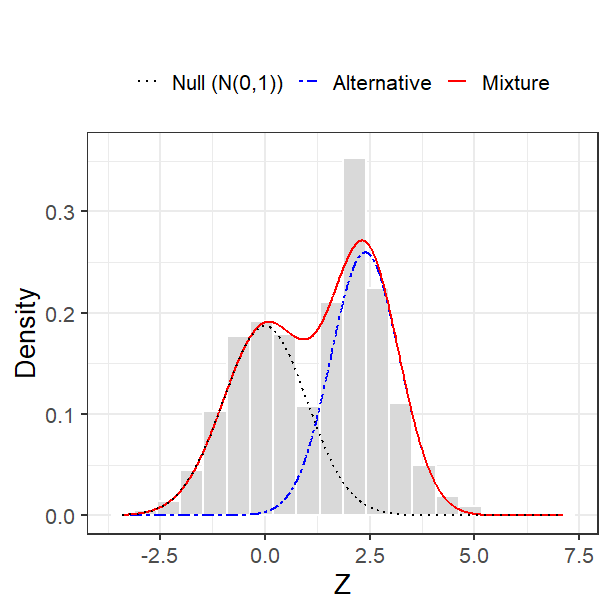}
}
\subfigure[Nonparametric method]{
\includegraphics[width=0.31 \linewidth]{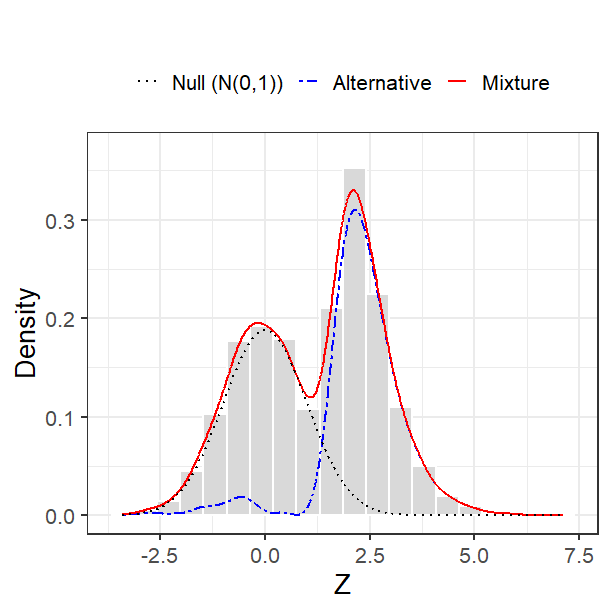}
}
\subfigure[Semiparametric method]{
\includegraphics[width=0.31 \linewidth]{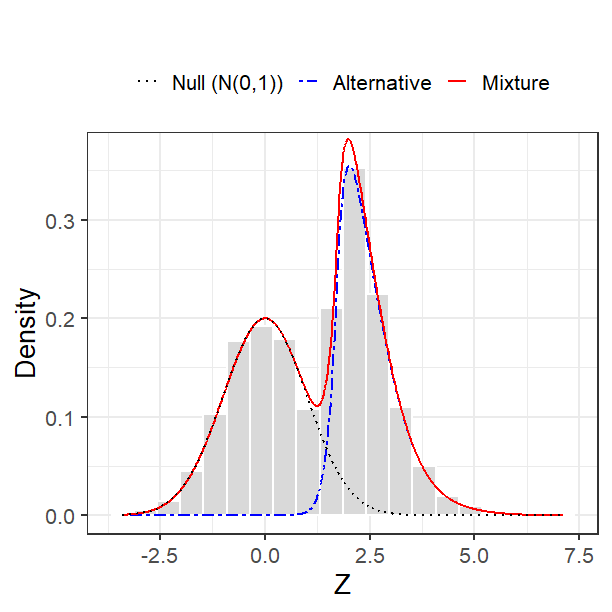}
}

\caption{Estimated density from each method with one simulated sample of Case $\uppercase\expandafter{\romannumeral6}$ when $N = 5000$}
\label{fig: Case 6}
\end{figure}

\section{Real Data Analysis} \label{sec:real}

In real gene–expression studies, the true null or non–null status of each gene is unknown, 
and therefore external clustering metrics such as ARI or AMI cannot be computed.  
Instead, following the empirical Bayes framework of \cite{efron2001empirical} and \cite{efron2002empirical}, 
we assess performance through posterior–probability–based error summaries.
Under the two–component mixture model in~\eqref{twocomp}, the local FDR 
is defined as the posterior probability that a gene with transformed statistic $z$ belongs to the null component, as given in~\eqref{localFDR}.  

For each gene $i$, we denote by $\hat{\gamma}_i$ the estimated posterior null probability obtained from the respective method.
Given a threshold $c$, we declare gene $i$ to be non–null whenever 
$\hat{\gamma}_i \le c$.  
Let
\[
N_r = \sum_{i=1}^N I_{[0,c]}(\hat{\gamma}_i)
\]
denote the total number of selected genes.  
Following \citet{mclachlan2006simple}, the estimated FDR is
\begin{equation*}
\widehat{\mathrm{FDR}}(c)
=
\frac{\displaystyle \sum_{i=1}^N \hat{\gamma}_j \, I_{[0,c]}(\hat{\gamma}_i)}
{\displaystyle N_r},
\end{equation*}
which estimates the expected proportion of null genes among the declared discoveries.
Conversely, the false negative rate (FNR) measures the proportion of truly non–null genes 
that fail to be selected.  
Its empirical Bayes estimator is
\begin{equation*}
\widehat{\mathrm{FNR}}(c)
=
\frac{\displaystyle \sum_{i=1}^N (1-\hat{\gamma}_i)\, I_{(c,\infty)}(\hat{\gamma}_i)}
{\displaystyle \sum_{i=1}^N (1-\hat{\gamma}_i)},
\end{equation*}
which averages the posterior probabilities of being non–null over the unselected genes.

These posterior–probability–based summaries provide practical and interpretable tools for comparing the competing mixture–modeling approaches in real datasets, where the true null and non–null labels are unobserved. In this context, methods achieving lower values of FDR and FNR are considered to exhibit superior differential–expression detection performance.

\subsection{Colon Cancer Data} \label{sec:real1}

\citet{alon1999broad} obtained microarray gene–expression measurements from colon cancer and normal colon tissues using Affymetrix oligonucleotide arrays. The dataset contains expression levels for more than 6500 genes across 40 tumor samples and 22 normal samples. Following the filtering strategy used in \citet{mclachlan2002mixture}, \citet{mclachlan2006simple}, and \citet{pommeret2019semiparametric}, we restrict the analysis to the $N=2000$ genes with the highest minimal intensity across samples to reduce the impact of low–expression noise.

Using the preprocessing procedure described in Section~\ref{sec:preprocessing}, we compute for each gene a pooled two–sample $t$–statistic comparing tumor and normal tissues, and then transform these statistics into $z$–scores via the mapping in~\eqref{eq:ztransform}. Genes yielding $z=-\infty$ which occurs when the observed test statistic provides essentially no evidence against the null hypothesis are excluded from downstream analysis. After this filtering step, a total of 971 genes remain.
For these transformed $z$–scores, we apply the three mixture–modeling approaches examined in the simulation study—the parametric method, the nonparametric method, and the proposed semiparametric method—to identify differentially expressed genes and to compare their behavior on real data.

Figures~\ref{fig: alon_fdr}--\ref{fig: alon_hist} summarize the results. Across all thresholds $c$, the proposed semiparametric method yields the lowest FDR and FNR, indicating more reliable identification of differentially expressed genes. The parametric and nonparametric methods show differing behavior: the parametric Gaussian mixture achieves uniformly lower FDR than the nonparametric estimator, whereas the FNR curves exhibit a threshold-dependent pattern. For $c < 0.3$, the parametric method has smaller FNR, while for $c > 0.3$ the nonparametric method performs better.

The estimated alternative densities in Figure~\ref{fig: alon_hist} provide additional insight. The nonparametric density is irregular and unstable, while the parametric model enforces a symmetric normal form that limits its ability to capture subtle departures from normality. The proposed semiparametric method produces a smooth and stable estimate that aligns closely with the empirical distribution, suggesting that the underlying alternative distribution in this dataset is neither strongly skewed nor heavy-tailed. 

\begin{figure}[ht] 
\centering

\subfigure[FDR]{
\includegraphics[width=0.45 \linewidth]{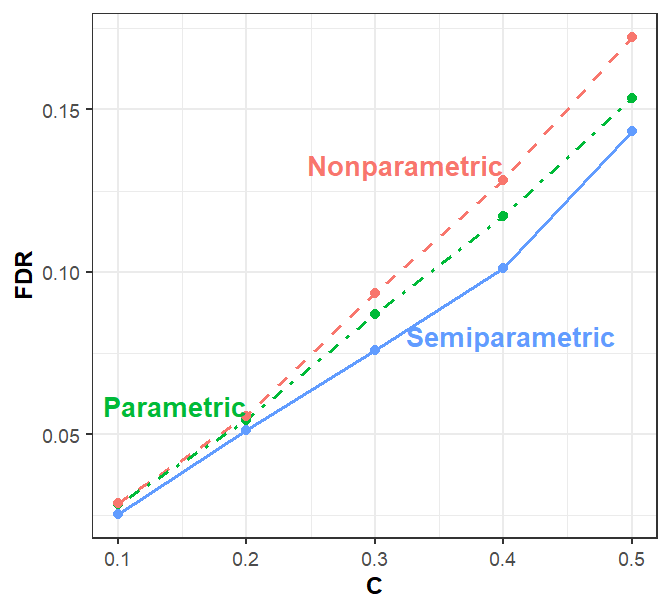}
}
\subfigure[FNR]{
\includegraphics[width=0.45 \linewidth]{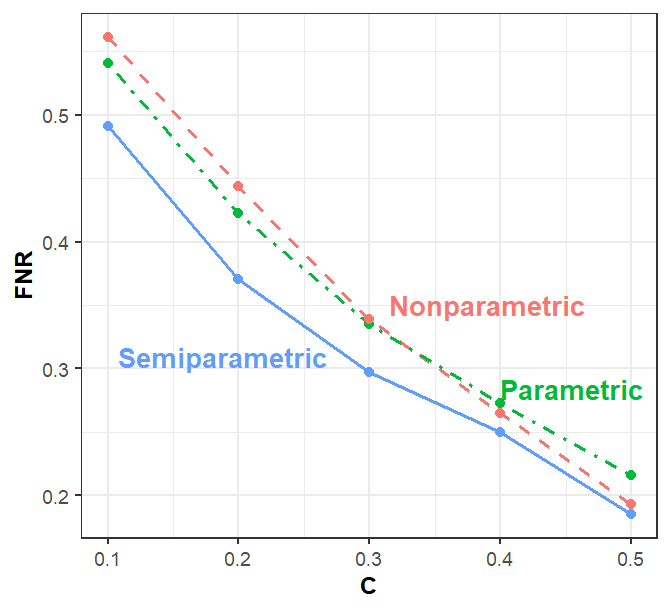}
}
\caption{Estimated FDR and FNR for the colon cancer dataset across a range of thresholds}
\label{fig: alon_fdr}
\end{figure}

\begin{figure}[ht] 
\centering

\subfigure[Parametric method]{
\includegraphics[width=0.31 \linewidth]{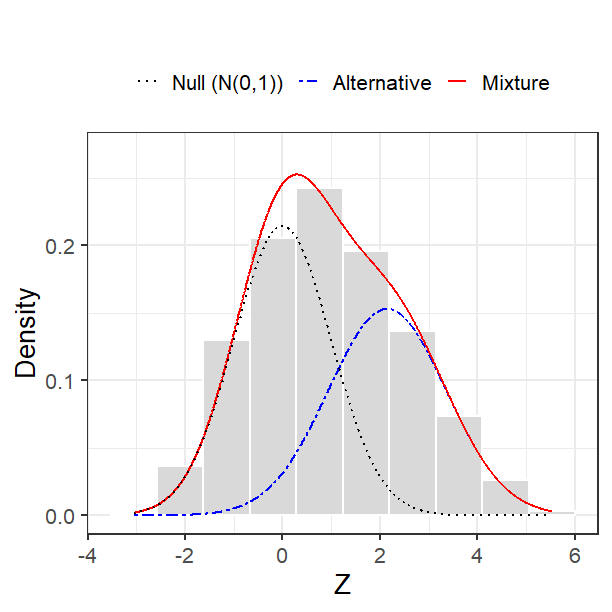}
}
\subfigure[Nonparametric method]{
\includegraphics[width=0.31 \linewidth]{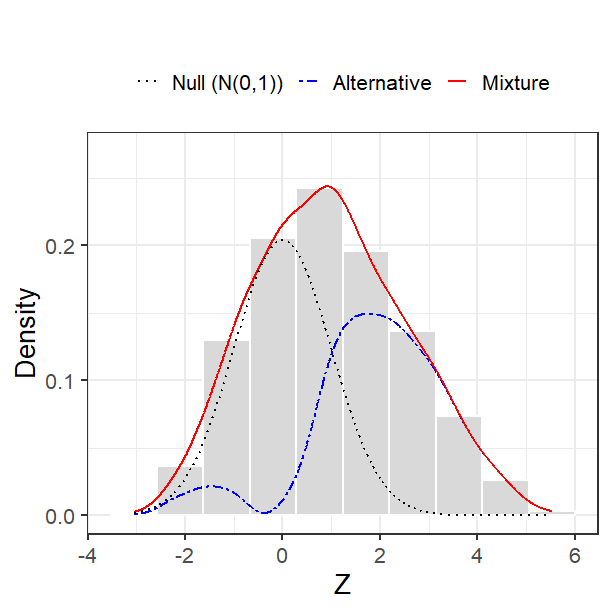}
}
\subfigure[Semiparametric method]{
\includegraphics[width=0.31 \linewidth]{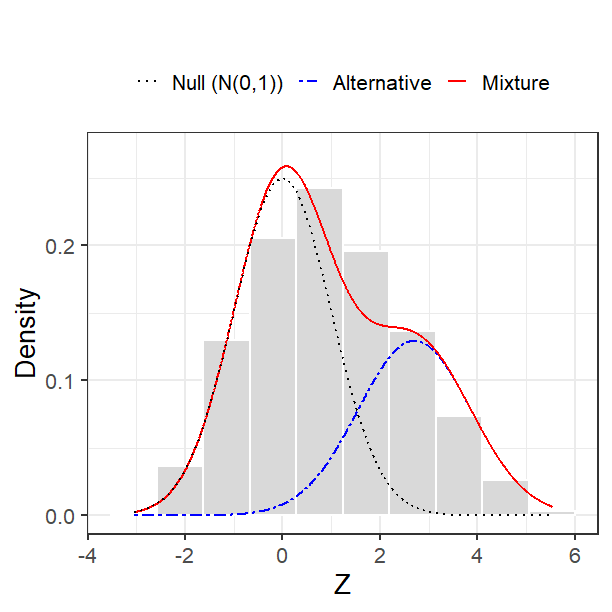}
}

\caption{Estimated null, alternative, and mixture densities for the colon cancer dataset}
\label{fig: alon_hist}
\end{figure}

\subsection{Leukemia Data} \label{sec:real2}

\citet{golub1999molecular} analyzed microarray gene–expression profiles from patients with acute leukemia, distinguishing Acute Lymphoblastic Leukemia (ALL) from Acute Myeloid Leukemia (AML) using Affymetrix high–density oligonucleotide arrays. The dataset consists of $M=72$ tissue samples (47 ALL and 25 AML) and $N=7129$ genes. Following the preprocessing pipeline of \citet{dudoit2002comparison}, we apply filtering to remove uninformative genes and then perform log–transformation and column–wise and row–wise standardization, yielding 3731 retained genes. Using the procedure in Section~\ref{sec:preprocessing}, we further restrict attention to 1853 informative genes. For these transformed $z$–scores, we fit the parametric, nonparametric, and proposed semiparametric mixture models to compute posterior probabilities and estimate FDR and FNR.

Figures~\ref{fig: golub_fdr}--\ref{fig: golub_hist} display the results for the leukemia dataset. For the FDR curves, all three methods behave similarly when the threshold $c$ is small, but clear differences emerge once $c$ exceeds approximately $0.2$. Beyond this point, the proposed semiparametric method achieves the lowest FDR across the entire range, while the parametric method yields slightly lower FDR than the nonparametric approach. A different pattern is observed for the FNR: for relatively large thresholds ($c  > 0.4$), the nonparametric method attains the smallest FNR, whereas for more stringent thresholds ($c < 0.4$), the proposed semiparametric method achieves the lowest FNR, outperforming both competitors. Taken together, these results indicate that the proposed method provides the most favorable balance between false discoveries and false negatives across practically relevant choices of $c$.

The estimated alternative densities in Figure~\ref{fig: golub_hist} further clarify these differences. The nonparametric estimator shows substantial variability, particularly in the tails, indicating sensitivity to sampling fluctuations, whereas the parametric Gaussian mixture enforces a symmetric normal form that cannot accommodate the evident asymmetry and heavy–tailed behavior present in the data. By contrast, the proposed semiparametric approach yields a smooth and stable density estimate that successfully captures both skewness and tail heaviness, consistent with the more flexible structure of the SNSM formulation.

\begin{figure}[ht] 
\centering

\subfigure[FDR]{
\includegraphics[width=0.45 \linewidth]{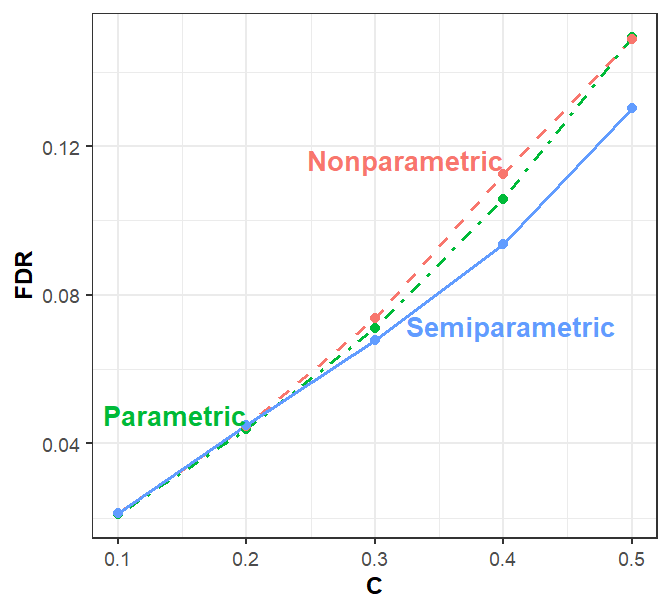}
}
\subfigure[FNR]{
\includegraphics[width=0.45 \linewidth]{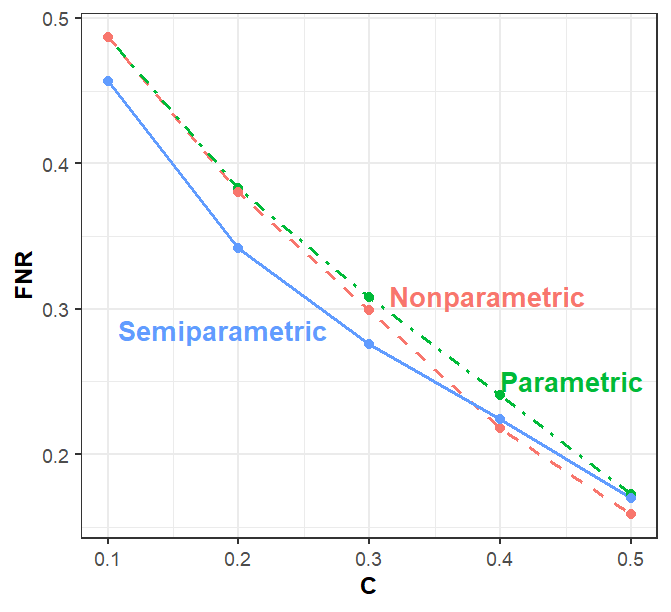}
}
\caption{Estimated FDR and FNR for the leukemia dataset across a range of thresholds}
\label{fig: golub_fdr}
\end{figure}

\begin{figure}[ht] 
\centering

\subfigure[Parametric method]{
\includegraphics[width=0.31 \linewidth]{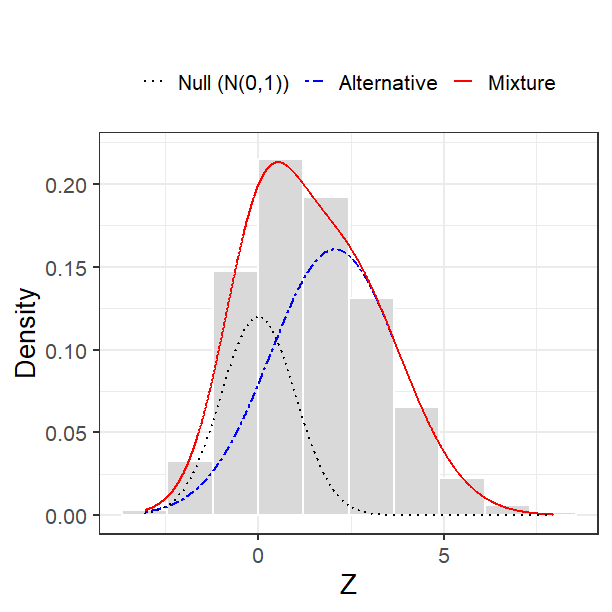}
}
\subfigure[Nonparametric method]{
\includegraphics[width=0.31 \linewidth]{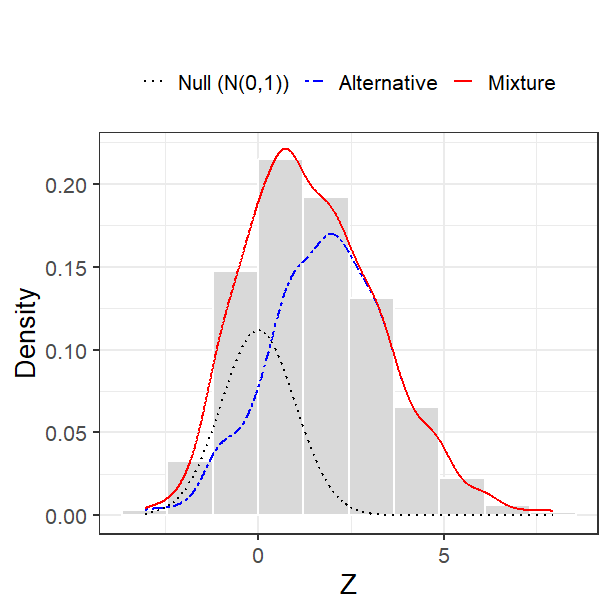}
}
\subfigure[Semiparametric method]{
\includegraphics[width=0.31 \linewidth]{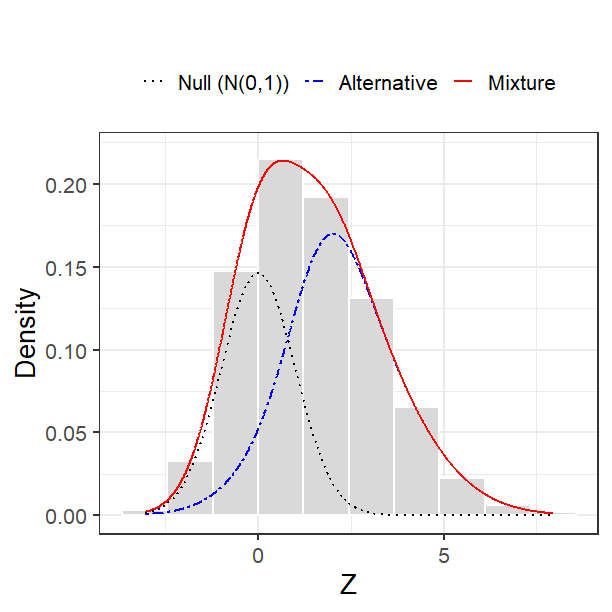}
}
\caption{Estimated null, alternative, and mixture densities for the leukemia dataset}
\label{fig: golub_hist}
\end{figure}

\section{Discussion} \label{sec:disc}

This paper introduces a semiparametric two-component mixture model for large--scale inference in gene-expression studies, in which the null distribution is fixed as standard normal and the alternative component is modeled by a semiparametric skew-normal scale mixture with an unspecified scale mixing distribution. By combining a finite--dimensional parametrization for location and skewness with a nonparametric maximum likelihood estimator for the scale distribution, the proposed model provides substantially greater flexibility than fully parametric approaches while avoiding the instability that often affects fully nonparametric methods. 

The numerical study demonstrates that the proposed method achieves consistently higher accuracy in separating null and non--null genes across a wide range of symmetric, skewed, light--tailed, and heavy--tailed alternatives. When the parametric Gaussian mixture is misspecified, the performance gains of the proposed method are particularly pronounced, and it also avoids the instability inherent in fully nonparametric estimators. Analyses of the colon cancer and leukemia datasets further support these findings: the proposed approach yields lower false discovery and false negative rates across practically relevant thresholds, and the estimated alternative densities exhibit moderate deviations from normality that are captured without unnecessary complexity. Overall, the proposed semiparametric mixture model provides a flexible, robust, and computationally tractable framework for differential gene--expression analysis, offering improved performance without imposing restrictive distributional assumptions.

As a direction for future work, it would be of interest to relax the assumption of a standard normal null distribution. Allowing for an empirical or nonparametric null component, for example under mild symmetry constraints, within the proposed framework could further improve robustness in identifying differential gene expression while maintaining identifiability and computational tractability. Such an extension would broaden the applicability of the proposed model beyond the strictly specified marginal setting considered here.

It would also be valuable to incorporate covariate information directly into the mixture structure. In many gene-expression studies, gene-specific or sample-specific covariates may provide additional information relevant to differential expression. Extending the proposed semiparametric alternative to covariate-dependent mixture formulations (e.g., \citealp{nguyen2016laplace}; \citealp{oh2023merging}; \citealp{oh2024semiparametric}; \citealp{chamroukhi2024functional}; \citealp{hwang2025mixture}) may further enhance flexibility and detection power while preserving the stability of the alternative density specification. Such developments would allow the framework to accommodate structured heterogeneity beyond the marginal mixture setting considered here.







\bibliographystyle{apa} 

\bibliography{references}

\section*{Appendix A: Proof of Theorem \ref{thm1}}
Assume that there exist $(\tilde{\pi}, \tilde{\mu}, \tilde{\lambda}, \tilde{G})$ such that
\begin{align}
\pi \phi(z; 0,1) + (1 - \pi) f_{\text{SNSM}}(z; \mu, \lambda, G) 
= \tilde{\pi} \phi(z; 0,1) + (1 - \tilde{\pi}) f_{\text{SNSM}}(z; \tilde{\mu}, \tilde{\lambda}, \tilde{G}). 
\label{ident}
\end{align}
 Rearranging terms gives
\[
(\pi-\tilde{\pi})\phi(z;0,1)
=
(1-\tilde{\pi})f_{\mathrm{SNSM}}(z;\tilde{\mu},\tilde{\lambda},\tilde{G})
-
(1-\pi)f_{\mathrm{SNSM}}(z;\mu,\lambda,G).
\]
Since both
\[
f_{\mathrm{SNSM}}(\cdot;\mu,\lambda,G),\quad
f_{\mathrm{SNSM}}(\cdot;\tilde{\mu},\tilde{\lambda},\tilde{G})
\in \mathcal{F}_{\mathrm{SNSM}}^{+},
\]
Assumption~\textup{(A1)} implies that the coefficient of
$\phi(\cdot;0,1)$ must be zero. Hence
\[
\pi=\tilde{\pi}.
\]

Taking characteristic functions on both sides of \eqref{ident}, we obtain:
\begin{align*}
\psi_z(t) = \pi \exp(-t^2/2) + (1 - \pi) \exp(i \mu t) \psi_{G, \lambda}(t),
\end{align*}
and
\begin{align*}
\tilde{\psi}_z(t) = \tilde{\pi} \exp(-t^2/2) + (1 - \tilde{\pi}) \exp(i \tilde{\mu} t) \psi_{\tilde{G}, \tilde{\lambda}}(t),
\end{align*}
where $i = \sqrt{-1}$ and $\psi_{G, \lambda}(t)$ denotes the characteristic function of the zero-location SNSM distribution with  $(\lambda, G)$.

Since $\pi = \tilde{\pi}$, equating both characteristic functions yields
\begin{equation}
\label{phase_eq}
\exp(i \mu t) \psi_{G,\lambda}(t)
=
\exp(i \tilde{\mu} t) \psi_{\tilde{G},\tilde{\lambda}}(t),
\qquad \forall t\in\mathbb{R}.
\end{equation}
Since $\psi_{G,\lambda}(0)=\psi_{\tilde{G},\tilde{\lambda}}(0)=1$ and both characteristic
functions are continuous, there exists $\delta>0$ such that
$\psi_{G,\lambda}(t)\neq 0$ and $\psi_{\tilde{G},\tilde{\lambda}}(t)\neq 0$ for all $|t|<\delta$.
Thus, for $|t|<\delta$, from \eqref{phase_eq} we obtain
\begin{equation}
\label{ratio_phase}
\frac{\psi_{G,\lambda}(t)}{\psi_{\tilde{G},\tilde{\lambda}}(t)}
=
\exp\!\{i(\tilde{\mu}-\mu)t\}.
\end{equation}
The right-hand side of \eqref{ratio_phase} is the characteristic function of a point-mass
distribution at $(\tilde{\mu}-\mu)$, that is, a pure location shift. Hence,
\eqref{ratio_phase} implies that if
$X \sim \mathrm{SNSM}(0,\lambda,G)$ and
$Y \sim \mathrm{SNSM}(0,\tilde{\lambda},\tilde{G})$, then
$X \stackrel{d}{=} Y + (\tilde{\mu}-\mu)$.
In particular, the zero-location parameterization is fixed by construction (i.e., the
location parameter is not absorbed into $(\lambda,G)$). Therefore, two distributions in the
subfamily $\{\mathrm{SNSM}(0,\lambda,G)\}$ cannot differ by a nonzero translation. Hence,
we must have $\tilde{\mu}-\mu=0$, that is,
\[
\mu=\tilde{\mu}.
\]

Since $\mu=\tilde{\mu}$, it follows from \eqref{phase_eq} that
\[
\psi_{G,\lambda}(t)=\psi_{\tilde{G},\tilde{\lambda}}(t),
\qquad \forall t\in\mathbb{R}.
\]
We now prove that $\psi_{G, \lambda}(t) = \psi_{\tilde{G}, \tilde{\lambda}}(t)$ implies $G = \tilde{G}$ and $\lambda = \tilde{\lambda}$. Note that
\begin{align*}
\psi_{G, \lambda}(t) 
&= \int \exp\left(-\frac{t^2 \sigma^2}{2}\right) dG(\sigma) 
+ i \int \exp\left(-\frac{t^2 \sigma^2}{2}\right) e\left( \frac{\Gamma(\lambda) \sigma t}{\sqrt{2}} \right) dG(\sigma), \\
\psi_{\tilde{G}, \tilde{\lambda}}(t) 
&= \int \exp\left(-\frac{t^2 \sigma^2}{2}\right) d\tilde{G}(\sigma) 
+ i \int \exp\left(-\frac{t^2 \sigma^2}{2}\right) e\left( \frac{\Gamma(\tilde{\lambda}) \sigma t}{\sqrt{2}} \right) d\tilde{G}(\sigma),
\end{align*}
where $\Gamma(\lambda) = \lambda / \sqrt{1 + \lambda^2}$ and $e(z)$ denotes the complementary error function:
\[
e(z) = -\frac{2i}{\sqrt{\pi}} \int_0^{iz} \exp(-t^2) dt.
\]
Equating the real and imaginary parts of $\psi_{G, \lambda}(t)$ and $\psi_{\tilde{G}, \tilde{\lambda}}(t)$, we obtain:
\begin{align}
\int \exp\left(-\frac{t^2 \sigma^2}{2}\right) dG(\sigma) &= \int \exp\left(-\frac{t^2 \sigma^2}{2}\right) d\tilde{G}(\sigma), \label{first} \\
\int \exp\left(-\frac{t^2 \sigma^2}{2}\right) e\left( \frac{\Gamma(\lambda) \sigma t}{\sqrt{2}} \right) dG(\sigma) 
&= \int \exp\left(-\frac{t^2 \sigma^2}{2}\right) e\left( \frac{\Gamma(\tilde{\lambda}) \sigma t}{\sqrt{2}} \right) d\tilde{G}(\sigma). \label{second}
\end{align}
Equation \eqref{first} implies $G = \tilde{G}$ by the uniqueness of the Laplace transform. Substituting into \eqref{second} and using the injectivity of $e(\cdot)$, we obtain $\Gamma (\lambda) = \Gamma (\tilde{\lambda})$, and hence $\lambda = \tilde{\lambda}$.
Therefore, identifiability holds.

\section*{Appendix B: Proof of Theorem \ref{thm2}}

Let us define the following metric:
\[
d \left( (\pi, \mu, \lambda, G), (\hat{\pi}, \hat{\mu}, \hat{\lambda}, \hat{G}) \right)
= | \pi - \hat{\pi} | + | \mu - \hat{\mu} | + | \lambda - \hat{\lambda} |
+ \int | G(\sigma) - \hat{G}(\sigma) | e^{-|\sigma|} \, d\tau(\sigma),
\]
where \(\tau\) denotes the Lebesgue measure on \(\mathbb{R}^+\).

\citet{kiefer1956consistency} established that under five assumptions—including continuity (Assumption 2), identifiability (Assumption 4), and integrability (Assumption 5)—the MLE converges in probability:
\[
d \left( (\pi, \mu, \lambda, G), (\hat{\pi}, \hat{\mu}, \hat{\lambda}, \hat{G}) \right) \overset{p}{\to} 0.
\]
The semiparametric SNSM density in \eqref{eq:proposed_model} trivially satisfies Assumptions 1, 2, and 3. Assumption 4 (identifiability) is established in Theorem~\ref{thm1}.

Since the support of \(G\) is restricted to \([\ell, \infty)\), the density in \eqref{eq:proposed_model} is uniformly bounded. Thus, it remains to verify Assumption 5, which requires:
\[
- \mathbb{E} \left[\log \left\{ \pi \phi(z) + (1 - \pi) f_{\text{SNSM}}(z; \mu, \lambda, G) \right\} \right] < \infty.
\]
Because \(\pi \in (0, 1)\), we have \(-\log \pi < \infty\) and \(-\log (1 - \pi) < \infty\). Moreover, the normal and skew-normal densities have finite first moments, implying \(-\mathbb{E}[\log \phi(z)] < \infty\) and 
\[
\int_{-\infty}^{\infty} f(z ; \mu, \lambda, \sigma) \left[\log |z| \right]^+ dz < \infty.
\]
By the condition \(\int_{\ell}^{\infty} \log \sigma \, dG(\sigma) < \infty\), it follows that 
\[
\mathbb{E} \left[\log |Z| \right]^+ < \infty,
\]
and thus
\[
- \int_{-\infty}^{\infty} \log \left\{ f(z ; \mu, \lambda, G) \right\} f(z ; \mu, \lambda, G) dz < \infty,
\]
by Lemma in Section 2 of \citet{kiefer1956consistency}. This verifies Assumption 5 and completes the proof.

\end{document}